\documentclass{article}
\usepackage{authblk}
\usepackage{amssymb, amsmath,amsthm}
\usepackage{graphicx}
\usepackage{url}
\usepackage{color}
\usepackage{xspace}
\usepackage{mathtools}
\usepackage{enumerate}
\usepackage{algorithm}
\usepackage{algorithmicx}
\usepackage{tikz}
\usepackage[color=green!20]{todonotes}
\usepackage{comment}

\newcommand{\yes}{\textsl{yes}}
\newcommand{\no}{\textsl{no}}
\newcommand{\suchthat}{\colon}
\newcommand{\poe}{{\mathsf{PoE}}}
\newcommand{\ptime}{{\mathsf{P}}}
\newcommand{\np}{{\mathsf{NP}}} 
\newcommand{\npo}{\mathsf{NPO}}

\newcommand{\fpt}{{\mathsf{FPT}}}   
\newcommand{\xp}{{\mathsf{XP}}}
\newcommand{\wone}{\mathsf{W[1]}}

\newcommand{\Oh}{\mathcal{O}}
\newcommand{\MCI}{\textsc{Multicolored Independent Set}\xspace}

\usepackage{tikz}
\newcommand{\SV}[1]{}
\newcommand{\LV}[1]{#1}

\newtheorem{proposition}{\bf Proposition}
\newtheorem{theorem}{\bf Theorem}

\newtheorem{example}{\bf Example}
\newtheorem{lemma}{\bf Lemma}
\newtheorem{corollary}{\bf Corollary}
\newtheorem{remark}{\bf Remark}
\newcommand{\longbib}[1]{}

\SV{\excludecomment{pf}}

\LV{
}

\tikzstyle{vertex}=[circle, draw, inner sep=1.2pt, minimum width=4pt, minimum size=0.4cm]
\tikzstyle{vertex2}=[circle, draw, inner sep=0pt, minimum width=4pt, minimum size=0.15cm]
\usetikzlibrary{decorations,decorations.pathmorphing,decorations.pathreplacing,fit,positioning}

\title{Extension of vertex cover and independent set in some classes of graphs and generalizations}

\author[1,2]{Katrin Casel}
\affil[1]{Hasso Plattner Institute, University of Potsdam,\newline  Potsdam, Germany}
\author[2]{Henning Fernau}
\affil[2]{Universit\"at Trier, Fachbereich 4,
Informatikwissenschaften, \newline
\hspace*{2cm}  Trier, Germany \newline
{\small\{casel,fernau\}@informatik.uni-trier.de}}
\author[3]{Mehdi Khosravian Ghadikolaei}
\affil[3]{Universit\'e Paris-Dauphine, PSL University,\newline CNRS, LAMSADE, 75016 Paris, France \newline
{\small{\{mehdi.khosravian-ghadikolaei,jerome.monnot, florian.sikora\}}@lamsade.dauphine.fr}}
\author[3]{J\'er\^ome Monnot}
\author[3]{Florian Sikora}

\newlength{\btw}
\newlength{\stw}
\newsavebox\tmpbox 
\date{}

\tikzstyle{vertex}=[circle, draw, inner sep=2pt, minimum width=4pt, minimum size=0.3cm]
\usetikzlibrary{decorations,decorations.pathmorphing,decorations.pathreplacing,fit,positioning}

\begin{document}
\maketitle

\begin{abstract}
We \LV{consider}\SV{study} extension variants of the classical\LV{ graph}  problems \textsc{Vertex\linebreak[2] Cover} and \textsc{Independent Set}.
Given a graph $G=(V,E)$ and a vertex set $U \subseteq V$, it is asked if there exists a \textit{minimal} vertex cover (resp.\ \textit{maximal} independent set) $S$ with $U\subseteq S$ (resp.\ $U \supseteq S$).  
Possibly contradicting intuition, these problems tend to be $\np$-hard, even in graph classes  where the classical problem can be solved \LV{in polynomial time}\SV{efficiently}. 
Yet, we exhibit some graph classes where the extension variant remains polynomial-time solvable.
We also study the parameterized complexity of theses problems, with parameter $|U|$, as well as the optimality of simple exact algorithms under \LV{the Exponential-Time Hypothesis}\SV{ETH}. 
All these complexity considerations are also carried out in very restricted scenarios, be it degree or topological restrictions (bipartite, planar or chordal graphs). This also motivates presenting some explicit branching algorithms for degree-bounded instances.

We further discuss the \textit{price of extension}, measuring the distance of $U$ to the closest set that can be extended, which results in natural optimization problems related to extension problems for which we discuss polynomial-time approximability.

\end{abstract}

\section{Introduction}

We will consider \textit{extension problems} related to the classical graph problems  \textsc{Vertex Cover} and \textsc{Independent Set}.
Informally in the extension version of  \textsc{Vertex Cover},  the input consists of 
both a graph~$G$ and a subset~$U$ of vertices, and the task is to extend~$U$ to an inclusion-wise minimal vertex cover  of~$G$  (if possible). With \textsc{Independent Set}, given a graph $G$  and a subset $U$ of vertices, we are looking for an inclusion-wise maximal independent set of $G$ contained in $U$.

Studying such version is interesting when one wants to develop efficient enumeration algorithms or also for branching algorithms, to name two examples of a list of applications given in~\cite{CasFKMS2018}.

\paragraph{Related work}

In~\cite{doi:10.1080/10556789808805708}, it is shown that extension of partial solutions is $\np$-hard for computing prime implicants of the dual of a Boolean function; a problem which can also be seen as trying to find a minimal hitting set for the prime implicants of the input function. Interpreted in this way, the proof from~\cite{doi:10.1080/10556789808805708} yields  $\np$-hardness for the minimal extension problem for {\sc 3-Hitting Set}. This result was extended in \cite{BazBCFJKLLMP2018} to prove  $\np$-hardness for the extension of minimal dominating sets (\textsc{Ext DS}), even restricted to planar cubic graphs.Similarly, it was shown in \cite{BazganBCF16}
that extension for minimum vertex cover restricted to planar cubic graphs is $\np$-hard.
The first \emph{systematic} study of this type of problems was exhibited in~\cite{CasFKMS2018} providing quite a number of different examples of this type of problem.

An \textit{independent system} is a set system $(V,\mathcal{E})$, $\mathcal{E}\subseteq 2^{V}$, that is hereditary under inclusion.
The extension problem \textsc{Ext Ind Sys} (also called \textsc{Flashlight}) for independent system was proposed in~\cite{LawLenKan80}. 
In this problem, given as input $X,Y\subseteq V$, one asks for the existence of a maximal independent set including $X$ and that does not intersect with $Y$. 
Lawler et al.\ proved that \textsc{Ext Ind Sys} is $\np$-complete, even when $Y=\emptyset$~\cite{LawLenKan80}. 
In order to enumerate all (inclusionwise) minimal dominating sets of a given graph, Kant\'e et al.\ studied a restriction of \textsc{Ext Ind Sys}: finding a minimal dominating set containing $X$ but excluding~$Y$. 
They proved that \textsc{Ext DS} is $\np$-complete, even in special graph classes like split graphs, chordal graphs and line graphs \cite{kante2015polynomial,kante2015polynomia}. 
Moreover, they proposed a linear algorithm for split graphs when $X,Y$ is a partition of the clique part~\cite{kante2014enumeration}.

\paragraph{Organization of the paper}

After some definitions and first results in Section~\ref{sec:def}, we  focus on bipartite graphs in Section~\ref{sec:bip} and give hardness results holding with strong degree or planarity constraints.
We also study parameterized complexity at the end of this section and comment on lower bound results based on ETH. 
In Section~\ref{sec:chordal}, we give positive algorithmic results on chordal graphs, with  a  combinatorial characterization for the subclass of trees.
We introduce the novel concept of \emph{price of extension} in Section~\ref{sec:PoE} and discuss  (non-)approximability for the according optimization problems. 
In Section~\ref{sec:hfree}, we generalize our results to $H$-free graphs for some fixed $H$. 
In Section~\ref{sec:bounded degree}, we prove several algorithmic results for bounded-degree graphs, based on a list of reduction rules and simple branching. 
Finally, in Section~\ref{sec:conclusions}, we give some prospects of future research.

\section{Definitions and preliminary results}\label{sec:def}

Throughout this paper, we consider simple undirected graphs only, to which we refer as \emph{graphs} henceforth. A graph can be specified by the set $V$ of vertices and the set $E$ of edges; every edge has two endpoints, and if $v$ is an endpoint of~$e$, we also say that $e$ and $v$ are incident. Let $G=(V,E)$ be a graph and $U\subseteq V$; $N_G(U)=\{v\in V\colon vu\in E\}$ denotes {\em the neighborhood} of $U$ in $G$ and $N_G[U]=U\cup N_G(U)$ denotes {\em the closed neighborhood} of $U$.
For singleton sets $U=\{u\}$, we simply write $N_G(u)$ or $N_G[u]$, even omitting $G$ if clear from the context. The cardinality of $N_G(u)$ is called {\em degree} of $u$, denoted $d_G(u)$. A graph where all vertices have  degree $k$ is called {\em $k$-regular}; 3-regular graphs are called {\em cubic}. If three upper-bounds the degree of all vertices we speak of {\em subcubic graphs}.\par
A vertex set $U$ induces the graph $G[U]$ 
with vertex set $U$ and $e\in E$ being an edge in $G[U]$ iff  both endpoints of $e$ are in $U$.
A vertex set $U$ is called \emph{independent} if $U\cap N_G(U)=\emptyset$; $U$ is called \emph{dominating} if $N_G[U]=V$; $U$ is  a \emph{vertex cover} if each edge $e$ is incident to at least one vertex from $U$.
A graph is called \emph{bipartite} if its vertex set decomposes into two independent sets. 
A vertex cover~$S$ is {\em minimal} if any proper subset $S^\prime \subset S$ of $S$ is not a vertex cover.
Clearly, a vertex cover $S$ is minimal iff each vertex $v$ in $S$ possesses a \emph{private edge}, i.e.,
an edge $vu$ with $u\notin S$. 
An independent set $S$ is {\em maximal} if any proper superset $S^\prime \supset S$ of $S$ is not an independent set. The two main problems discussed in this paper are:

\begin{center}
\fbox{\begin{minipage}{.95\textwidth}
\noindent{\textsc{Ext VC}}\\\nopagebreak
{\bf Input:} A graph $G=(V,E)$, a set of vertices $U \subseteq V$. \\\nopagebreak
{\bf Question:}  Does $G$ have a minimal vertex cover $S$ with $U\subseteq S$?
\end{minipage}}
\end{center}

\begin{center}
\fbox{\begin{minipage}{.95\textwidth}
\noindent{\textsc{Ext IS}}\\\nopagebreak
{\bf Input:} A graph $G=(V,E)$, a set of vertices $U \subseteq V$. \\\nopagebreak
{\bf Question:}  Does $G$ have a maximal independent set $S$ with $S\subseteq U$?
\end{minipage}}
\end{center}

\noindent
For \textsc{Ext VC}, the set $U$ is also referred to as the set of \emph{required} vertices. As complements of maximal independent sets are minimal vertex covers we conclude:
\begin{remark}\label{rem-Ext VC=IS}
 $(G,U)$ is a \yes-instance of \textsc{Ext VC} iff $(G,V\setminus U)$ is a \yes-instance of \textsc{Ext IS}, as complements of maximal independent sets are minimal vertex covers. 
\end{remark}


 Since adding or deleting edges between vertices of $U$ does not change the minimality of feasible solutions of  \textsc{Ext VC}, 
 we can first state the following.

\begin{remark}\label{obsExt_vc}
For \textsc{Ext VC} (and for \textsc{Ext IS}) one can always assume the required vertex set (the set $V\setminus U$) is either a clique or an independent set.
\end{remark}

The following theorem gives a combinatorial characterization of \yes-instances of \textsc{Ext VC} that is quite important in our subsequent discussions. \SV{Its proof is moved to the appendix. This is notified here (and similarly below) by adding~$(*)$.}

\begin{theorem}\label{caract_Ext_VC} \SV{$(*)$}
Let $G=(V,E)$ be a graph and  $U \subseteq V$ be a set of vertices. The three following conditions are equivalent:

\begin{description}

\item[$(i)$] $(G,U)$ is a \yes-instance  of \textsc{Ext VC}.

\item[$(ii)$] $(G[{N_G[U]}],N_G[U]\setminus U)$ is a \yes-instance  of \textsc{Ext IS}.
\item[$(iii)$] There exists an independent dominating set $S^\prime\subseteq N_G[U]\setminus U$ of $G[{N_G[U]}]$.
\end{description}
\end{theorem}

\newcommand{\proofofcaractExtVC}{\begin{proof} In the following arguments, let $G=(V,E)$ be a graph. 
Let us first look at conditions $(ii)$ and $(iii)$.
By our previous discussions,  
condition $(ii)$ is equivalent to: $(G[N_G[U]],U)$  is a \yes-instance  of \textsc{Ext VC}.
Assume there is a minimal vertex cover $S$ of $G[N_G[U]]$ with $U\subseteq S$. Hence, in particular we deduce $N_G[v]\nsubseteq U$ for every $v\in U$ by minimality of $S$. Condition $(ii)$ therefore entails
the  existence of an independent set $S'$ of $G[{N_G[U]}]$ with $S'\subseteq (N_G[U]\setminus U)$ and $U\subseteq N_G[S']$. Hence, condition $(ii)$ implies condition $(iii)$.
Conversely, let $S^\prime\subseteq N_G[U]\setminus U$ be  an independent dominating set of $G[{N_G[U]}]$.
Clearly, $S=N_G[U]\setminus S'$ is a vertex cover of $G[{N_G[U]}]$. If $S$ were not minimal, then there would be a vertex $v\in S$ with $N_{G[N_G[U]]}(v)\subseteq S$, as then $v$ would not possess a private edge. But then $v$ would not be dominated by any vertex from $S'$, violating the assumption that $S'$ is a dominating set of $G[N_G[U]]$.
Hence, conditions $(ii)$ and $(iii)$ are equivalent. 

Now, we will prove the equivalence between items $(i)$ and $(iii)$.
Let $S$ be a minimal vertex cover of $G$ with $U\subseteq S$. Clearly, $S\cap  N_G[U]$ is a vertex cover of $G[N_G[U]]$, but notice that it need not be minimal, as private edges of $v\in S\cap  N_G[U]$ need not lie in the graph induced by $N_G[U]$. 
The set $S^\prime=(V\setminus S)\cap N_G[U]\subseteq N_G[U]\setminus U$ is an independent set (as the complement of $S\cap  N_G[U]$ within $G[N_G[U]]$) which dominates all the vertices in $U$. Namely, consider any $u\in U$ and assume that $u\notin N_G[S']$. Then, $N_G[u]\subseteq S$, contradicting minimality of $S$.
We turn $S^\prime$ into a maximal independent set of the induced graph $G[{N_G[U]\setminus U}]$, by adding some vertices from ${N_G[U]\setminus U}$ to $S'$. Observe that the resulting set $S''$ is also a maximal independent set in $G[{N_G[U]}]$ and hence
%
satisfies condition $(iii)$, because each $u\in U$ has a private edge (as being part of the minimal vertex cover $S$ of $G$, connecting $u$ to some  $v\in S^\prime$.
Conversely, assume the existence of an independent dominating set $R$ of $G$ satisfying $(iii)$. Hence, $R$ is an independent set with $R\subseteq (N_G[U]\setminus U)$ and $U\subseteq N_G[R]$. 
Let $X$ be any maximal independents set of $G[V\setminus N_G[R]]$, for instance, produced by some greedy procedure. Let  $S^\prime:=R\cup X$. By construction, $S'$ is an independent set in $G$. If $S'$ were not maximal, then we would find some $x\in S'$ with $N_G(x)\cap S'= \emptyset$. Clearly, $x\notin N_G[R]$.  But as $x$ has no neighbors in $X$, it could have been added to $X$ by the mentioned greedy procedure. In conclusion, $S'$ is a maximal independent set.  
Hence, $S=V\setminus S^\prime$ satisfies the condition $(i)$.
\end{proof}}

\LV{\proofofcaractExtVC}


\section{Bipartite graphs}\label{sec:bip}

In this section, we focus on bipartite graphs.
We prove that  \textsc{Ext VC} is  
$\np$-complete, even if restricted to cubic, or planar subcubic graphs. 
Due to Remark~\ref{rem-Ext VC=IS}, this immediately yields the same type of results for  \textsc{Ext IS}.
We add some algorithmic notes on planar graphs that are also valid for the non-bipartite case. Also, we discuss results based on ETH. 
We conclude the section by studying the parameterized complexity of \textsc{Ext VC} in bipartite graphs when parameterized by the size of~$U$.

\begin{theorem}\label{Bip_Ext_VC}
\textsc{Ext VC} (and \textsc{Ext IS}) is $\np$-complete in cubic bipartite graphs.
\end{theorem}
\begin{proof}
We reduce from {\sc 2-balanced 3-SAT}, denoted \textsc{$(3,B2)$-SAT}\SV{, which is  $\np$-hard by~\cite[Theorem~1]{ECCC-TR03-049}}, where an instance $I$
is given by a set 
$C$ of CNF clauses over a set
$X$ 
of Boolean variables such that each clause
has exactly $3$ literals and  each variable appears exactly $4$ times,
twice negative and twice positive.
The bipartite graph associated to 
$I$ 
is\LV{ the graph} $BP=(C\cup X,E(BP))$ with $C=\{c_1,\dots,c_m\}$, $X=\{x_1,\dots,x_n\}$ and $E(BP)=\{c_jx_i\colon x_i$ or $\neg  x_i$ is literal of $c_j\}$.
\LV{Deciding whether an instance of \textsc{$(3,B2)$-SAT} is satisfiable 
is $\np$-complete by~\cite[Theorem~1]{ECCC-TR03-049}.}\par
For an instance $I=(C,X)$ of \textsc{$(3,B2)$-SAT}, 
we build a cubic bipartite graph $G=(V,E)$ by duplicating instance $I$  (here, vertices $C'=\{c'_1,\dots,c'_m\}$ and $X'=\{x'_1,\dots,x'_n\}$ are the duplicate variants of vertices $C=\{c_1,\dots,c_m\}$ and $X=\{x_1,\dots,x_n\}$) and by connecting gadgets as done in Figure~\ref{Fig:ExtVC_CUB_BIP}.
We also add the following edges between the two copies: $l_il'_i$,  $m_im'_i$ and $r_ir'_i$ for $i=1,\dots,n$. The construction is illustrated in  Figure~\ref{Fig:ExtVC_CUB_BIP} and clearly, $G$ is a cubic bipartite graph. Finally we set $U=\{c_i,c'_i\colon i=1,\dots,m\}\cup \{m_j,m'_j\colon j=1,\dots,n\}$.

\begin{figure}[t]
\centering
\begin{tikzpicture}[scale=0.5, transform shape]
\tikzstyle{vertex}=[circle, draw, inner sep=0pt, inner sep=0pt, minimum size=0.8cm]

\node[vertex,ultra thick] (c1) at (0,0) {$c_1$};
\node[vertex, below of=c1, node distance=1.6cm,ultra thick] (c2) {$c_2$};
\node[vertex, below of=c2, node distance=1.6cm,ultra thick] (c3) {$c_3$};
\node[below of=c3,node distance=1cm] (cdots) {$\vdots$};
\node[vertex,below of=cdots,ultra thick] (cm) {$c_m$};

\node[vertex] (x1) at (3,0.5) {$x_1$};
\node[vertex, below of=x1, node distance=1.2 cm] (nx1p) {$\neg x_1^{\prime}$};
\node[vertex, below of=nx1p, node distance=1 cm] (x2) {$x_2$};
\node[vertex, below of=x2, node distance=1.2 cm] (nx2p) {$\neg x_2^{\prime}$};
\node[below of=nx2p] (xdots) {$\vdots$};
\node[vertex, below of=xdots] (xn) {$x_n$};
\node[vertex, below of=xn, node distance=1.2 cm] (nxnp) {$\neg x_n^{\prime}$};

\node[vertex, right of=x1, node distance=2cm] (l1) {$l_1$};
\node[vertex, below of=l1, node distance=1.2 cm] (l1p) {$l_1^{\prime}$};
\node[vertex, below of=l1p, node distance=1 cm] (l2) {$l_2$};
\node[vertex, below of=l2, node distance=1.2 cm] (l2p) {$l_2^{\prime}$};
\node[below of=l2p] (ldots) {$\vdots$};
\node[vertex, below of=ldots] (ln) {$l_n$};
\node[vertex, below of=ln, node distance=1.2 cm] (lnp) {$l_n^{\prime}$};

\node[vertex, right of=l1, node distance=2cm,ultra thick] (m1) {$m_1$};
\node[vertex, below of=m1, node distance=1.2 cm,ultra thick] (m1p) {$m_1^{\prime}$};
\node[vertex, below of=m1p, node distance=1 cm,ultra thick] (m2) {$m_2$};
\node[vertex, below of=m2, node distance=1.2 cm,ultra thick] (m2p) {$m_2^{\prime}$};
\node[below of=m2p] (mdots) {$\vdots$};
\node[vertex, below of=mdots,ultra thick] (mn) {$m_n$};
\node[vertex, below of=mn, node distance=1.2 cm,ultra thick] (mnp) {$m_n^{\prime}$};

\node[vertex, right of=m1, node distance=2cm] (r1) {$r_1$};
\node[vertex, below of=r1, node distance=1.2 cm] (r1p) {$r_1^{\prime}$};
\node[vertex, below of=r1p, node distance=1 cm] (r2) {$r_2$};
\node[vertex, below of=r2, node distance=1.2 cm] (r2p) {$r_2^{\prime}$};
\node[below of=r2p] (rdots) {$\vdots$};
\node[vertex, below of=rdots] (rn) {$r_n$};
\node[vertex, below of=rn, node distance=1.2 cm] (rnp) {$r_n^{\prime}$};

\node[vertex, right of=r1, node distance=2cm] (nx1) {$\neg x_1$};
\node[vertex, below of=nx1, node distance=1.2 cm] (x1p) {$x_1^{\prime}$};
\node[vertex, below of=x1p, node distance=1 cm] (nx2) {$\neg x_2$};
\node[vertex, below of=nx2, node distance=1.2 cm] (x2p) {$x_2^{\prime}$};

\node[below of=x2p] (nxdots) {$\vdots$};
\node[vertex, below of=nxdots] (nxn) {$\neg x_n$};
\node[vertex, below of=nxn, node distance=1.2 cm] (xnp) {$x_n^{\prime}$};

\node[vertex,ultra thick] (c1p) at (14,0.2) {$c_1^{\prime}$};
\node[vertex, below of=c1p, node distance=1.6 cm,ultra thick] (c2p) {$c_2^{\prime}$};
\node[vertex, below of=c2p, node distance=1.6 cm,ultra thick] (c3p) {$c_3^{\prime}$};
\node[below of=c3p,node distance=1cm] (cpdots) {$\vdots$};
\node[vertex,below of=cpdots,ultra thick] (cmp) {$c_m^{\prime}$};

\draw (x1) -- (l1) -- (m1) -- (r1) -- (nx1);
\draw (nx1p) -- (l1p) -- (m1p) -- (r1p) -- (x1p);
\draw (x2) -- (l2) -- (m2) -- (r2) -- (nx2);
\draw (nx2p) -- (l2p) -- (m2p) -- (r2p) -- (x2p);
\draw (xn) -- (ln) -- (mn) -- (rn) -- (nxn);
\draw (nxnp) -- (lnp) -- (mnp) -- (rnp) -- (xnp);
\draw (l1)--(l1p);
\draw (l2)--(l2p);
\draw (ln)--(lnp);
\draw (m1)--(m1p);
\draw (m2)--(m2p);
\draw (mn)--(mnp);
\draw (r1)--(r1p);
\draw (r2)--(r2p);
\draw (rn)--(rnp);

\draw (c1) edge[bend left=20] (nx1);
\draw (c1p) edge[bend right=47] (nx1p);
\draw (c1) edge[] (xn);
\draw (c1p) edge[] (xnp);
\draw (x1)--(c2)--(x2)--(cm);
\draw (x1p)--(c2p)--(x2p)--(cmp);

\draw (cm) edge[bend right=50] (nxn);
\draw (cmp) edge[bend left=32] (nxnp);

\end{tikzpicture}
 \caption{Graph $G=(V,E)$ for \textsc{Ext VC} built from $I$. Vertices of $U$ have a bold border.}
 \label{Fig:ExtVC_CUB_BIP}
\end{figure}
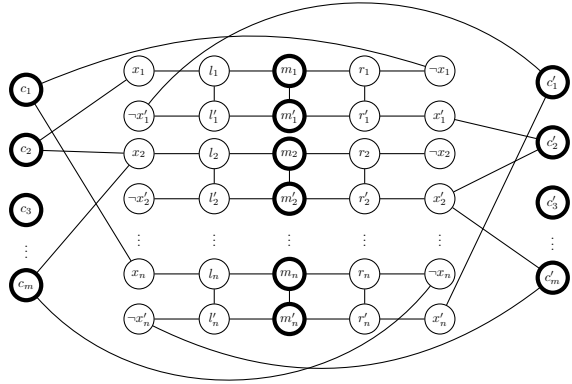

\noindent
We claim that $I$ is satisfiable iff $G$ admits a minimal vertex cover containing~$U$.

\noindent Assume $I$ is satisfiable and let $T$ be a truth assignment which satisfies all clauses. We set $S=\{\neg x_i,l_i,\neg x'_i,r'_i\colon T(x_i)=\textit{true}\}\cup \{x_i,r_i,x'_i,l'_i\colon T(x_i)=\textit{false}\}\cup U$. We can easily check that $S$ is a minimal vertex cover containing $U$.

\noindent Conversely, assume that $G$ possesses a minimal vertex cover $S$ containing $U$. In order to cover the edges $l_il'_i$ and $r_ir'_i$,  for every $i=1,\dots,n$, either the set of two vertices $\{l_i,r'_i\}$ or  $\{l'_i,r_i\}$ belongs to $S$. Actually, for a fixed $i$, we know that $|\{l_i,l'_i,r_i,r'_i\}\cap S|\geq 2$; if $\{l_i,l'_i\}\subseteq S$ or $\{r_i,r'_i\}\subseteq S$, then $S$ is not a minimal vertex cover, because $m_i$ or $m'_i$ can be deleted, which is a contradiction. Hence, if $\{l_i,r'_i\}\subseteq S$ (resp., $\{r_i,l'_i\}\subseteq S$), then $\{\neg x_i,\neg x'_i\}\subseteq S$ (resp., $\{x_i,x'_i\}\subseteq S$), since the edges $l'_i\neg x'_i$ and $r_i\neg x_i$ (resp., $l_ix_i$ and $r'_ix_i$) must be covered. In conclusion, by setting $T(x_i)=\textit{true}$ if $\neg x_i\in S$ and $T(x_i)=\textit{false}$ if $x_i\in S$ we obtain a truth assignment $T$ which satisfies all clauses, because $\{C_i,C'_i\colon i=1,\dots,m\}\subseteq U\subseteq S$.
\end{proof}



In the following, we discuss restriction to planar graphs.
\LV{In order to prove our results, we will present reductions from} We use the problem \textsc{4-Bounded Planar 3-Connected SAT} (\textsc{4P3C3SAT} for short), the restriction of {\sc 3-satisfiability}\LV{ to clauses in $C$ 
over variables in 
$X$,} 
where each variable occurs in at most four clauses (at least one time negative and one time positive) and the associated bipartite\LV{ {\em variable-gadget}} graph $BP$  is planar\LV{ of maximum degree~4}. This restriction is also  $\np$-complete~\cite{Kra94}.
\smallskip

Let $I=(X,C)$ be an instance of \textsc{4P3C3SAT}, where $X=\{x_1,\dots,x_n\}$ and $C=\{c_1,\dots,c_m\}$ are variable and clause sets of $I$, respectively. By definition, the graph $G=(V,E)$ with $V=\{c_1,\dots,c_m\}\cup \{x_1,\dots,x_n\}$ and $E=\{c_ix_j\colon x_j \text{ or }\neg x_j \text{ appears in } c_j\}$ is planar. In the following, we always assume that the planar graph comes with an embedding in the plane.
Informally, we are building  a new graph by putting some gadgets instead of vertices $x_i$ of $G$ which satisfy two following conditions: (1) as it can be seen in Fig.~\ref{Fig:ExtVC_CUB_BIP}, the constructions distinguishes between the cases that a variable $x_i$ appears positively and negatively in some clauses (2) the construction preserves planarity.
\smallskip

Suppose that variable $x_i$ appears in the clauses $c_1,c_2,c_3,c_4$ of instance $I$ such that in the induced (embedded) subgraph $G_i=G[\{x_i,c_1,c_2,c_3,c_4\}]$, $c_1x_i$, $c_2x_i$, $c_3x_i$, $c_4x_i$ is an anti-clockwise ordering of edges around $x_i$. By looking at $G_i$ and considering  $x_i$ appears positively and negatively,
the construction should satisfy one of the  following cases:\LV{
\pagebreak[2]
\begin{itemize}
\item[$\bullet$]} case~1: $x_i\in c_1, c_2$ and $\neg x_i \in c_3,c_4$;\LV{\item[$\bullet$]} case~2: $x_i\in c_1,c_3$ and $\neg x_i \in c_2,c_4$;\LV{\item[$\bullet$]} case~3: $x_i\in c_1,c_2,c_3$ and $\neg x_i \in c_4$.
\LV{\end{itemize}}
Note that all other cases are included in these by rotations\LV{ or replacing $x_i $ with $\neg x_i$ or vice versa.}\SV{ or swapping the roles of $x_i $ and $\neg x_i$.}

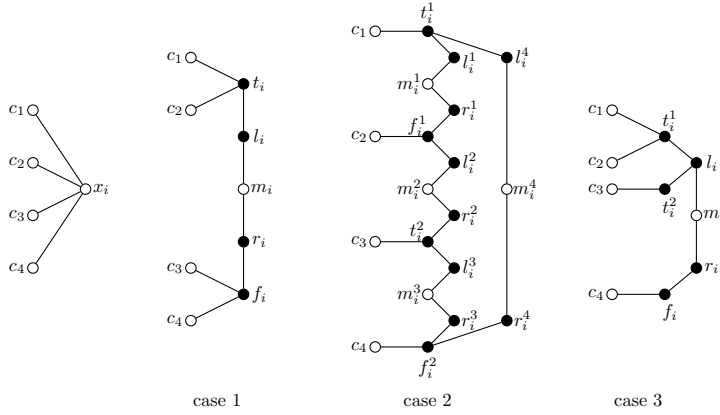
\begin{figure}[t]
\centering
\begin{tikzpicture}[scale=0.7, transform shape]
\tikzstyle{vertex1}=[circle, draw, inner sep=2pt,  minimum width=1 pt, minimum size=0.1cm]
\tikzstyle{vertex}=[circle, draw, inner sep=2pt, fill=black!100, minimum width=1pt, minimum size=0.1cm]

\node[vertex1] (x) at (-4,-3) {};
\node[vertex1] (c1) at (-5,-1.5) {};
\node[vertex1] (c2) at (-5,-2.5) {};
\node[vertex1] (c3) at (-5,-3.5) {};
\node[vertex1] (c4) at (-5,-4.5) {};
\draw (x)--(c1);
\draw (x)--(c2);
\draw (x)--(c3);
\draw (x)--(c4);
\node () at (-3.7,-3) {$x_i$};
\node () at (-5.3,-1.5) {$c_1$};
\node () at (-5.3,-2.5) {$c_2$};
\node () at (-5.3,-3.5) {$c_3$};
\node () at (-5.3,-4.5) {$c_4$};

\node[vertex] (21) at (-1,-1) {};
\node[vertex,below of=21,node distance=1cm](22){};
\node[vertex1,below of=22,node distance=1cm](23){};
\node[vertex,below of=23,node distance=1cm](24){};
\node[vertex,below of=24,node distance=1cm](25){};
\node[vertex1] (2c1) at (-2,-0.5) {};
\node[vertex1] (2c2) at (-2,-1.5) {};
\node[vertex1] (2c3) at (-2,-4.5) {};
\node[vertex1] (2c4) at (-2,-5.5) {};
\draw (21)--(22)--(23)--(24)--(25);
\draw (21)--(2c1);
\draw (21)--(2c2);
\draw (25)--(2c3);
\draw (25)--(2c4);

\node () at (-0.7,-1) {$t_i$};
\node () at (-0.7,-2) {$l_i$};
\node () at (-0.65,-3) {$m_i$};
\node () at (-0.7,-4) {$r_i$};
\node () at (-0.7,-5) {$f_i$};
\node () at (-2.3,-0.5) {$c_1$};
\node () at (-2.3,-1.5) {$c_2$};
\node () at (-2.3,-4.5) {$c_3$};
\node () at (-2.3,-5.5) {$c_4$};

\node () at (-1.5,-7) {case 1};

\node[vertex] (11) at (2.5,0) {};
\node[vertex](12) at (3,-0.5){};
\node[vertex1](13) at (2.5,-1){};
\node[vertex](14) at (3,-1.5){};
\node[vertex](15) at (2.5,-2){};
\node[vertex](16) at (3,-2.5){};
\node[vertex1](17) at (2.5,-3){};
\node[vertex](18) at (3,-3.5){};
\node[vertex](19) at (2.5,-4){};
\node[vertex](110) at (3,-4.5){};
\node[vertex1](111) at (2.5,-5){};
\node[vertex](112) at (3,-5.5){};
\node[vertex](113) at (2.5,-6){};
\node[vertex] (114) at (4,-0.5) {};
\node[vertex1] (115) at (4,-3) {};
\node[vertex] (116) at (4,-5.5) {};
\node[vertex1,left of=11,node distance=1 cm](1c1){};
\node[vertex1,left of=15,node distance=1 cm](1c2){};
\node[vertex1,left of=19,node distance=1 cm](1c3){};
\node[vertex1,left of=113,node distance=1 cm](1c4){};

\node () at (2.5,0.35) {$t_i^{1}$};
\node () at (3.3,-0.6) {$l_i^{1}$};
\node () at (2.15,-1) {$m_i^{1}$};
\node () at (3.3,-1.5) {$r_i^{1}$};
\node () at (2.3,-1.8) {$f_i^{1}$};
\node () at (3.3,-2.5) {$l_i^{2}$};
\node () at (2.15,-3) {$m_i^{2}$};
\node () at (3.3,-3.5) {$r_i^{2}$};
\node () at (2.3,-3.8) {$t_i^{2}$};
\node () at (3.3,-4.5) {$l_i^{3}$};
\node () at (2.15,-5) {$m_i^{3}$};
\node () at (3.3,-5.5) {$r_i^{3}$};
\node () at (2.5,-6.4) {$f_i^{2}$};
\node () at (4.3,-0.5) {$l_i^{4}$};
\node () at (4.35,-3) {$m_i^{4}$};
\node () at (4.3,-5.5) {$r_i^{4}$};

\node () at (1.2,0) {$c_1$};
\node () at (1.2,-2) {$c_2$};
\node () at (1.2,-4) {$c_3$};
\node () at (1.2,-6) {$c_4$};

\node () at (2.5,-7) {case 2};

\draw (11)--(12)--(13)--(14)--(15)--(16)--(17)--(18)--(19)--(110)--(111)--(112)--(113)--(116)--(115)--(114)--(11);

\draw (11)--(1c1);
\draw (15)--(1c2);
\draw (19)--(1c3);
\draw (113)--(1c4);

\node[vertex] (31) at (7,-2) {};
\node[vertex] (32) at (7,-3) {};
\node[vertex] (33) at (7.6,-2.5) {};
\node[vertex1] (34) at (7.6,-3.5) {};
\node[vertex] (35) at (7.6,-4.5) {};
\node[vertex] (36) at (7,-5) {};

\node[vertex1] (3c1) at (6,-1.5) {};
\node[vertex1] (3c2) at (6,-2.5) {};
\node[vertex1, left of=32, node distance=1 cm] (3c3)  {};
\node[vertex1, left of=36, node distance=1 cm] (3c4)  {};

\node () at (7.1,-1.7) {$t_i^{1}$};
\node () at (7.1,-3.35) {$t_i^{2}$};
\node () at (7.1,-5.35) {$f_i$};
\node () at (7.9,-2.5) {$l_i$};
\node () at (7.95,-3.5) {$m_i$};
\node () at (7.9,-4.5) {$r_i$};
\node () at (5.7,-1.5) {$c_1$};
\node () at (5.7,-2.5) {$c_2$};
\node () at (5.7,-3) {$c_3$};
\node () at (5.7,-5) {$c_4$};
\node () at (6.5,-7) {case 3};

\draw (3c1)--(31)--(3c2);
\draw (3c3)--(32);
\draw (3c4)--(36);
\draw (31)--(33)--(34)--(35)--(36);
\draw (32)--(33);

\end{tikzpicture}
\caption{Construction of Theorem~\ref{EXT IS planar}. On the left: A variable $x_i$ appearing in four clauses $c_1,c_2,c_3,c_4$ in $I$. On the right, cases 1,2,3: The gadgets $H(x_i)$ in the constructed instance, depending on how $x_i$ appears (negative or positive) in the four clauses. Black vertices denote elements of~$U$.} \label{Fig:Ext IS}
\end{figure}

\begin{theorem}\label{EXT IS planar} \SV{$(*)$}
\textsc{Ext IS} is $\np$-complete on planar bipartite subcubic graphs.
\end{theorem}

\newcommand{\proofofExtISplanarshort}{
The proof is based on a reduction from \textsc{4P3C3SAT}. We start from graph $G$ which is defined already above and build a planar bipartite graph $H$ by replacing every node $x_i$ in $G$ with one of the three gadgets $H(x_i)$ which are depicted in Fig.~\ref{Fig:Ext IS}. Let \LV{$}$F_1=\{m_i\colon H(x_i) \text{ complies with cases 1 or 3}\}$\LV{$} and let \LV{$}$F_2=\{m_i^{1}, m_i^{2}, m_i^{3}, m_i^{4}\colon H(x_i) \text{ complies with case 2}\}\,.$\LV{$} The permitted vertex set is $U=V(H)\setminus (F_1\cup F_2 \cup C)$, where $C=\{c_i\colon 1\leq i\leq m\}$. This construction is  polynomial-time computable and $H$ is a planar bipartite subcubic graph. \LV{We claim that }$H$ has a maximal independent set which contains only vertices from $U$ iff $I$ is satisfiable.}

\newcommand{\proofofExtISplanarReasoning}{If $T$ is a truth assignment of $I$ which satisfies all clauses, then depending on $T(x_i)=\textit{true}$ or $T(x_i)=\textit{false}$, we define the independent set $S_i$ corresponding to three different variable gadgets $H(x_i)$ as follows:

$S_i: =\begin{cases}\{t_i,r_i\} & \text{if $H(x_i)$ adapts to case 1 and } T(x_i)=\textit{true},\\ \{t_i^{1},r_i^{1},l_i^{2},t_i^{2},r_i^{3},r_i^{4}\} & \text{if $H(x_i)$ adapts to case 2 and } T(x_i)=\textit{true},\\\{t_i^{1},t_i^{2},r_i\} & \text{if $H(x_i)$ adapts to case 3 and } T(x_i)=\textit{true},\\\{f_i,l_i\}& \text{if $H(x_i)$ adapts to case 1 and } T(x_i)=\textit{false},\\\{l_i^{1},f_i^{1},r_i^{2},l_i^{3},f_i^{2},l_i^{4}\} & \text{if $H(x_i)$ adapts to case 2 and } T(x_i)=\textit{false},\\\{l_i,f_i\}& \text{if $H(x_i)$ adapts to case 3 and } T(x_i)=\textit{false}.\end{cases}$
\smallskip

We can see that  $S=\bigcup_{1\leq i\leq n}S_i$ is a maximal independent set of $H$ which contains only vertices from $U$.
\smallskip

Conversely, suppose $S\subseteq U$ is a maximal independent set of $H$. By using maximality of $S$, we define an assignment $T$ for $I$ depending on different types of variable gadgets of $H$ as follows:
\begin{itemize}
\item[$\bullet$] for case 1, one of $l_i,r_i$ must be in $S$, hence we set $T(x_i)=\textit{true}$ (resp., $T(x_i)=\textit{false}$) if $r_i\in S$ (resp., $l_i\in S$).

\item[$\bullet$] for case 2, at least one of vertices in each pair $\{(l_i^{j},r_i^{j})\colon 1\leq j\leq 4\})$ must be in $S$. Hence, at most one of $(S\cap\{t_i^{1},t_i^{2}\})\neq \emptyset$ and $(S\cap\{f_i^{1},f_i^{2}\})\neq \emptyset$ is true. Thus we set $T(x_i)=\textit{true}$ (resp., $T(x_i)=\textit{false}$) if $(S\cap\{t_i^{1},t_i^{2}\})\neq \emptyset$ (resp., $(S\cap\{f_i^{1},f_i^{2}\})\neq \emptyset$).

\item[$\bullet$] for case 3, one can see, similar to the previous two cases: if one of $t_i^{1},t_i^{2}$ (resp., $f_i$) is in $S$, then none of $f_i$ (resp. $t_i^{1},t_i^{2}$) are in $S$, then we set $T(x_i)=\textit{true}$ (resp., $T(x_i)=\textit{false}$) if $(S\cap \{t_i^{1},t_i^{2}\})\neq \emptyset$ (resp., $f_i\in S$).
\end{itemize}

We obtain a valid assignment $T$. This assignment satisfies all clauses of $I$, since for all $c_j\in C$,
$(N(c_j)\cap S)\neq \emptyset$ (by maximality of $S$).}

\SV{\noindent\emph{Sketch.} \proofofExtISplanarshort}

\LV{\begin{proof}\proofofExtISplanarshort
\smallskip
\proofofExtISplanarReasoning
\end{proof}}

\SV{\smallskip}\noindent
It is challenging to strengthen \LV{the previous}\SV{this} result to planar bipartite cubic  graphs.


\subsubsection*{Algorithmic notes for the planar case}

By distinguishing between whether a vertex belongs to the cover or not and further, when it belongs to the cover, if it already has a private edge or not, it is not hard to design a dynamic programming algorithm 
that decides in time $\Oh^*(c^t)$ if $(G,U)$ is a \yes-instance of \textsc{Ext VC} or not, given a graph $G$ together with a tree decomposition of width $t$. With some more care, even $c=2$ can be achieved, but this is not so important here. Rather, below we will make explicit an algorithm on trees that is based on 
several combinatorial properties and hence differ from the DP approach sketched here for the more general notion of treewidth-bounded graphs.

Moreover, it is well-known that planar graphs of order $n$ have treewidth bounded by $\Oh(\sqrt{n})$.
In fact, we can obtain a corresponding tree decomposition in polynomial time, given a planar graph $G$.
Piecing things together, we obtain:

\begin{theorem}
\textsc{Ext VC} can be solved in time  $\Oh(2^{\Oh(\sqrt{n})})$ on planar graphs\LV{ of order $n$}.
\end{theorem}

\subsubsection*{Remarks on the Exponential Time Hypothesis}

 Assuming ETH, there is no $2^{o(n+m)}$-algorithm for solving $n$-variable, $m$-clause instances of 
$(3,B2)$-{\sc SAT}.
As our reduction from $(3,B2)$-{\sc SAT} increases the size of the instances only in a linear fashion, we can immediately conclude:

\begin{theorem}\label{thm:eth-summary}
 There is no $2^{o(n+m)}$-algorithm for  $n$-vertex, $m$-edge bipartite subcubic instances of $\textsc{Ext VC}$\LV{ or  $\textsc{Ext IS}$}, unless ETH fails.
\end{theorem}

This also motivates us to further study exact exponential-time algorithms. 
We can also deduce optimality of our algorithms for planar graphs based on the following auxiliary result.

\begin{proposition}\label{prop:SAT}
There is no algorithm that solves 
\textsc{4P3C3SAT} on instances with $n$ variables and $m$ clauses in time $2^{o(\sqrt{n+m})}$, unless ETH fails.
\end{proposition}

\begin{corollary}\label{cor:ETH-planar}
There is no $2^{o(\sqrt{n})}$ algorithm for solving \textsc{Ext VC} on planar instances of order~$n$, unless ETH fails.
\end{corollary}

\subsubsection*{Remarks on Parameterized Complexity} We now study our problems in the framework of parameterized complexity where we consider the size of the set of fixed vertices as \emph{standard parameter} for our extension problems.

\begin{theorem}\label{thm: W-c Ext VC}
\textsc{Ext VC} with standard parameter is $\wone$-complete, even when restricted to bipartite instances. 
\end{theorem}
\begin{proof}
We show hardness by reduction from \MCI,\SV{ which is $\wone$-hard by~\cite{DBLP:journals/tcs/FellowsHRV09},\footnote{The proof is  for \textsc{Multicolored Clique}; taking the complement graph is a parameterized reduction showing that \MCI is $\wone$-hard.}} so let $G=(V,E)$ be an instance of of this problem, with partition $V_1,\dots,V_{k}$ for $V$. W.l.o.g., assume that each $V_i$ induces a clique and $|V_i|\geq 2$.  Construct $G'=(V',E')$ from $G$ with $V'$ built from two copies of $V$, denoted $V$ and $\bar V:= \{\bar v\suchthat v\in V\}$, and $2k$ additional vertices $\{w_i,\bar w_i\suchthat 1\leq i\leq k\}$, and $E'$ containing $u\bar v$ for all $uv\in E$ and $u\bar w_i$ and $\bar uw_i$ for all $u\in V_i$, $i\in \{1,\dots, k\}$ (see Fig.~\ref{Fig:W_hard_Ext_VC}). $G'$ is bipartite with partition into $V\cup\{w_i\suchthat 1\leq i\leq k\}$ and  $\bar V\cup\{\bar w_i\suchthat 1\leq i\leq k\}$.  Set $U=\{w_i,\bar w_i\suchthat 1\leq i\leq k\}$ and consider  $(G', U)$  as instance of \textsc{Ext VC}.  We claim that $(G', U)$ is a \yes-instance for \textsc{Ext VC} iff $G$ is a \yes-instance for \MCI.\LV{  Since \MCI is $\wone$-hard~\cite{DBLP:journals/tcs/FellowsHRV09},\footnote{The proof is  for \textsc{Multicolored Clique}; taking the complement graph is a parameterized reduction showing that \MCI is $\wone$-hard.} this $\fpt$-reduction shows  $\wone$-hardness for \textsc{Ext VC} with standard parameterization.}\par
Suppose $(G', U)$ is a \yes-instance for \textsc{Ext VC}, so there exists a minimal vertex cover $S$ for $G'$ with $U\subseteq S$. Consider $S':= V'\setminus S$. Since $S$ is minimal,  $N(u)\not\subseteq S$ for all $u\in S$, so especially for each  $i\in\{1,\dots, k\}$ there exists at least one vertex from $N(w_i)=\bar V_i$ in $S'$ and also at least one vertex from $N(\bar w_i)=\bar V_i$. Since $S'$ has to be an independent set in $G'$ and $v\bar u\in E'$ for all $u,v\in V_i$, $u\not=v$ (recall that $V_i$ is a clique in $G$), it follows that if $v\in S'\cap V_i$, then $\bar v$ is the only vertex independent from $v$ in $\bar V_i$. This means that $|S'\cap V_i|= 1$ for all $i\in\{1,\dots, k\}$ and if $S'\cap V=\{v_1,\dots, v_k\}$, then $S'\cap \bar V=\{\bar v_1,\dots, \bar v_k\}$. The set $S'\cap V$ hence is a multicolored independent set in $G$, since $v_iv_j\in E$ for $i,j\in\{1,\dots,k\}$ would imply that $v_i\bar v_j\in E'$ which is not possible since $S'$ is an independent set in $G'$.
Conversely, it is not hard to see that if there exists a multicolored independent set $S$ in $G$, then the set $V'\setminus (S\cup \bar S)$ (with $\bar S:=\{\bar v\suchthat v\in S\}$) is a minimal vertex cover for $G'$ containing $U$.\par

\par
Membership in $\wone$ is seen as follows. As suggested in~\cite{Ces2003}, we describe a reduction to \textsc{Short Nondeterministic Turing Machine}. Given a graph $G=(V,E)$ and a pre-solution $U=\{u_1,\dots,u_k\}\subseteq V$, the constructed Turing machine first guesses vertices $u_1',\dots,u_k'$, with $u_i'\in N(u_i)\setminus U$ and then verifies in time $\Oh(k^2)$ if the guessed set $U'$ is an independent set. As $\{u_1',\dots,u_k'\}$ can be greedily extended to an independent dominating set for $N[U]$ which, by Theorem~\ref{caract_Ext_VC}, is equivalent to $(G,U)$ being a \yes-instance of {\sc Ext VC}, $U$ can be extended to a minimal vertex cover iff one of the guesses is successful.
\end{proof}

\begin{figure}[bt]
\tikzstyle{vertex1}=[circle, draw, inner sep=0pt, minimum width=4pt, minimum size=0.5cm]
\centering
\mbox{ }\\[-16ex]
\begin{tikzpicture}[scale=0.95, transform shape]
\node[vertex] (v1) at (0,0) {};
\node[vertex, right of=v1, node distance=0.5 cm] (v2) {};
\node[vertex, right of=v2, node distance=0.5 cm] (v3) {};
\node[right of=v3, node distance=1 cm] (hdots) {$\hdots$};
\node[vertex, right of=hdots, node distance=1 cm] (vn-1) {};
\node[vertex, right of=vn-1, node distance=0.5 cm] (vn) {};

\node[vertex, below of=v1, node distance=1.5 cm] (vb1) {};
\node[vertex, right of=vb1, node distance=0.5 cm] (vb2) {};
\node[vertex, right of=vb2, node distance=0.5 cm] (vb3) {};
\node[right of=vb3, node distance=1 cm] (hbdots) {$\hdots$};
\node[vertex, right of=hbdots, node distance=1 cm] (vbn-1) {};
\node[vertex, right of=vbn-1, node distance=0.5 cm] (vbn) {};

\draw (vb1)--(v3);
\draw(vb2)--(v3);
\draw(vb2)--(vn-1);
\draw(vb3)--(v1);
\draw(vb3)--(v2);
\draw(vb3)--(vn);
\draw(vbn-1)--(v2);
\draw(vbn-1)--(vn);
\draw(vbn)--(v3);
\draw(vbn)--(vn-1);

\node[draw, rectangle,rounded corners,label=above:$$,fit= (v1) (vb3)] (){};
\node[draw, rectangle,rounded corners,label=above:$$,fit= (vn-1) (vbn)] (){};
\node () at (4.1,0) {$V$};
\node () at (4.1,-1.5) {$\bar V$};

\node[vertex1,ultra thick] (w1) at (0.5,-2.5) {$w_1$};
\node[vertex1, below of=w1, node distance=1 cm, ultra thick] (wb1) {$\bar w_1$};

\node[vertex1,ultra thick] (wk) at (3.3,-2.5) {$w_k$};
\node[vertex1, below of=wk, node distance=1 cm, ultra thick] (wbk) {$\bar w_k$};
\draw(w1)--(wb1);
\draw(wk)--(wbk);
\node[below of=hbdots, node distance=1.8 cm] (hwdots) {$\hdots$};
\draw (w1)--(vb1);
\draw (w1)--(vb2);
\draw (w1)--(vb3);
\draw (wk)--(vbn-1);
\draw (wk)--(vbn);
\draw (0.3,-3.6) .. controls (-0.8,-4.5) and (-1.2,1.8) .. (0,0.12);
\draw (0.45,-3.7) .. controls (-1.2,-5) and (-1.7,2.8) .. (0.5,0.12);
\draw (0.56,-3.75) .. controls (-1.4,-5.7) and (-2.4,3.6) .. (1,0.12);
\draw (3.5,-3.6) .. controls (5,-5) and (5,2.5) .. (3.5,0.12);
\draw (3.4,-3.7) .. controls (5.5,-5.7) and (5.5,3.5) .. (2.95,0.12);
\end{tikzpicture}\\[-9.5ex]
\caption{The graph $G'=(V',E')$ for \textsc{Ext VC}, Vertices in $U$ are drawn bold.}\label{Fig:W_hard_Ext_VC}
\end{figure}
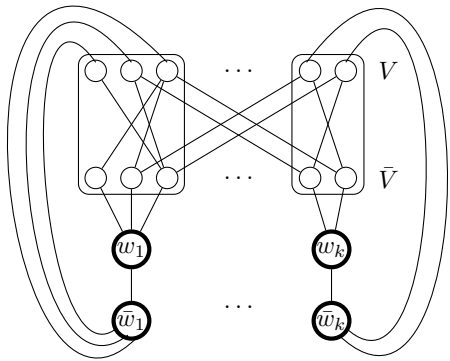

As a remark, it is obvious to see that considering the parameter $n-|U|$ instead of $|U|$ leads to an $\fpt$-result, as it is sufficient to test if any of the subsets of $V\setminus U$, together with $U$, form a minimal vertex cover. However, these algorithms are quite trivial and hence not further studied here.
The same reasoning shows:

\begin{remark}
\textsc{Ext IS} with standard parameter is in $\fpt$. 
\end{remark}


\begin{theorem}\label{thm-fpt-VC-planar}
\textsc{Ext VC} with standard parameter is in $\fpt$ on planar graphs.
\end{theorem}

\begin{proof} Let $(G,U)$ be an instance of \textsc{Ext VC} such that $G$ is planar.
By Theorem~\ref{caract_Ext_VC}, it suffices to solve  \textsc{Ext VC} on $(G',U)$, where $G'$ is the graph induced by $N_G[U]$. Clearly, $G'$ is also planar.
Moreover, the diameter of each connected component of $G'$ is upper-bounded by $3|U|$. Therefore, 
$G'$ is (at most) $3|U|$-outerplanar and hence according to~\cite{Bod96},  the treewidth of $G'$ is at most~$9|U|$. Our previous remarks show that  \textsc{Ext VC} can be solved in time $\Oh^*(2^{\Oh(|U|)})$.
\end{proof}


\section{Chordal and Circular-arc graphs} \label{sec:chordal}
An undirected graph $G=(V,E)$ is {\em chordal} iff each cycle of $G$ with a length at least four has a chord  (an edge linking two non-consecutive vertices of the cycle) and $G$ is {\em circular-arc} if it is the
intersection graph of a collection of $n$ arcs  around a circle. We will need the following problem definition.
\begin{center}
\fbox{\begin{minipage}{.95\textwidth}
\noindent{\textsc{Minimum Independent Dominating Set} (\textsc{MinISDS} for short)}\\\nopagebreak
{\bf Input:} A graph $G=(V,E)$. \\\nopagebreak
{\bf Solution:} Subset of vertices $S\subseteq V$ which is independent and dominating. \\\nopagebreak
{\bf Output:}  Solution $S$ that minimizes  $|S|$.
\end{minipage}}
\end{center}

\textsc{Weighted Minimum Independent Dominating Set} (or 
\textsc{WMinISDS} for short) corresponds to the vertex-weighted variant of  \textsc{MinISDS}, where each vertex $v\in V$ has a non-negative weight $w(v)\geq 0$ associated to it and the goal consists in minimizing $w(S)=\sum_{v\in S}w(v)$.
If $w(v)\in \{a,b\}$ with $0\leq a<b$, the weights are called {\em bivaluate},
and $a=0$ and $b=1$ corresponds to {\em binary weights}.

\begin{remark}\label{rem-WISDS}
\textsc{MinISDS} for chordal graphs has been studied
in \cite{Farber82ORL}, where it is shown that the restriction to binary weights is solvable in polynomial-time. Bivalued \textsc{MinISDS} with $a>0$ however is already $\np$-hard on chordal graphs, see~\cite{Chang04}. \textsc{WMinISDS} (without any restriction on the number of distinct weights) is also polynomial-time solvable in circular-arc graphs \cite{Chang98}. 
\end{remark}

\LV{Using the mentioned polynomial-time result of \textsc{binary independent dominating set} on chordal graphs \cite{Farber82ORL} and circular-arc graphs \cite{Chang98}, we deduce:}

\begin{corollary}\label{cor:chordal_Ext_VC}
\textsc{Ext VC} is polynomial-time decidable in chordal and in circular-arc graphs.
\end{corollary}
\begin{proof}
By Remark~\ref{rem-WISDS}, we can find, within polynomial-time, an independent dominating set $S^*$ minimizing $w(S^*)=\sum_{v\in S^*}w(v)$  among the independent dominating sets of a weighted chordal graph or circular-arc graph $(G,w)$ where $G=(V,E)$ and  $\forall v\in V$, $w(v)\in\{0,1\}$.

Let $(G,U)$ be an instance of \textsc{Ext VC} where $G=(V,E)$ is a chordal graph (resp., a circular-arc graph). We will apply the result of \cite{Farber82ORL} (resp., \cite{Chang98})  for $(G^\prime,w)$,
where $G^\prime$ is the subgraph of $G$ induced by $N_G[U]$ and $w(v)=1$ if $v\in U$ and $w(v)=0$ for $v\in N_G[U]\setminus U$. Obviously,
$(G^\prime,w)$ is a binary-weighted chordal graph (resp., circular-arc graph). So, an optimal independent dominating set $S^*$ of $(G^\prime,w)$ has a weight $0$ iff $S^*\subseteq N_G[U]\setminus U$ is a maximal independent set of $G^\prime$, otherwise $w(S^*)\geq 1$. Using
 Theorem~\ref{caract_Ext_VC},  the result follows. 
\end{proof}

Farber's algorithm \cite{Farber82ORL}  runs in linear-time and is based on the resolution of a linear programming using  primal and dual programs. 
Yet, it would be nice to find a (direct) combinatorial linear-time algorithm for chordal and circular-arc graphs, as this is quite common in that area. We give a first step in this direction by presenting a characterization of \yes-instances of {\sc Ext VC} on trees.\par
Consider a  tree $T=(V,E)$  and a set of vertices $U$. A subtree  $T^\prime=(V^\prime,E^\prime)$ (ie., a connected induced subgraph) of a tree $T$ is called \emph{edge full with respect to $(T,U)$} if $U\subseteq V^\prime$,
$d_{T^\prime}(u)=d_{T}(u)$ for all $u\in U$.
A subtree $T^\prime=(V^\prime,E^\prime)$  is \emph{induced edge full with respect to $(T,U)$} if
it is edge full with respect to $(T,U\cap V^\prime)$.\par
For our characterization, we use a coloring of vertices with colors black and white. If $T=(V,E)$ is a tree and $X\subseteq V$, we use $T[X\to\mathrm{black}]$ to denote the colored tree where exactly the vertices from $X$ are colored black.
Further define the following class of black and white colored trees $\mathcal{T}$,
inductively as follows. Base case: A tree with a single vertex $x$ belongs to  $\mathcal{T}$ if $x$ is black. Inductive step: If $T\in \mathcal{T}$, the tree resulting from the addition of a $P_3$ (3 new vertices that form a path $p$) where one endpoint of $p$ is black, the two other vertices are white and the white endpoint of $p$ is linked to a black vertex of $T$ is in $\mathcal{T}$.\par
\begin{theorem}\label{Tree_carac_Ext_VC}\SV{$(*)$}
Let $T=(V,E)$ be a tree and $U\subseteq V$ be an independent set. Then, $(T,U)$ is a \yes-instance of \textsc{Ext VC} \SV{iff}\LV{if and only if} there is no subtree $T^\prime=(V^\prime,E^\prime)$ of $T$ that is induced edge full with respect to $(T,U)$ such that $T^\prime[U\to\mathrm{black}]\in\mathcal{T}$.
\end{theorem}

\section{Price of extension}\label{sec:PoE}
Considering the possibility that some fixed set $U$ might not be extendible to any minimal solution, one might ask how wrong $U$ is as a fixed choice for an extension problem. One idea to evaluate this, is to ask how much $U$ has to be altered when aiming for a minimal solution. Described differently for our extension problems at hand, we want to discuss  how many vertices of $U$ have to be deleted for  \textsc{Ext VC} (added for  \textsc{Ext IS}) in order to arrive at a \yes-instance of the extension problem. The magnitude of how much $U$ has to be altered can be seen as the price that has to be paid to ensure extendibility.\par
In order to formally discuss this concept, we consider according optimization problems. From an instance $I=(G,U)$ of \textsc{Ext VC} or \textsc{Ext IS}, we define two new maximization (resp., minimization) $\npo$ problems, respectively.

\begin{center}
\fbox{\begin{minipage}{.95\textwidth}
\noindent{\textsc{Max Ext VC}}\\\nopagebreak
{\bf Input:} A graph $G=(V,E)$, a set of vertices $U \subseteq V$. \\\nopagebreak
{\bf Solutions:} Minimal vertex cover $S$ of $G$. \\\nopagebreak
{\bf Output:}  Solution $S$ that maximizes $|S\cap U|$.
\end{minipage}}
\end{center}

\begin{center}
\fbox{\begin{minipage}{.95\textwidth}
\noindent{\textsc{Min Ext IS}}\\\nopagebreak
{\bf Input:} A graph $G=(V,E)$, a set of vertices $U \subseteq V$. \\\nopagebreak
{\bf Solutions:} Maximal independent set $S$ of $G$. \\\nopagebreak
{\bf Output:}  Solution $S$ that minimizes  $|U|+|S\cap (V\setminus U)|$.
\end{minipage}}
\end{center}

\noindent For $\Pi=$\textsc{Max Ext VC} or \textsc{Min Ext IS}, we denote by $opt_{\Pi}(I,U)$  the value of an optimal solution of \textsc{Max Ext VC} or
\textsc{Min Ext IS}, respectively. Since for both of them, $opt_{\Pi}(I,U)=|U|$ iff $(G,U)$ is a \yes-instance of \textsc{Ext VC} or \textsc{Ext IS}, respectively, we deduce that \textsc{Max Ext VC} and
\textsc{Min Ext IS} are $\np$-hard as soon as \textsc{Ext VC} and \textsc{Ext IS} are  $\np$-complete. \SV{Alternatively, we could write}\LV{ 

Notice that alternatively these two optimal quantities can be expressed as}

\noindent
$opt_{\textsc{Max Ext VC}}(G,U)=\arg\max\{U^\prime\subseteq U\suchthat (G,U^\prime) \text{ is a \yes-instance of }\textsc{Ext VC}\}$\ and\LV{\\}
$opt_{\textsc{Min Ext IS}}(G,U)=\arg\min\{U^\prime\supseteq U\suchthat (G,U^\prime) \text{ is a \yes-instance of }\textsc{Ext IS}\}$.

Similarly to Remark~\ref{rem-Ext VC=IS}, one observes that
the decision variants of  \textsc{Max Ext VC} and \textsc{Min Ext IS} are\LV{indeed completely} equivalent, more precisely:

\begin{equation}\label{eq:Max Ext VC_Min Ext IS}
opt_{\textsc{Max Ext VC}}(G,U)+opt_{\textsc{Min Ext IS}}(G,U^\prime)=|V|\,.
\end{equation}

We want to discuss polynomial-time approximability of  \textsc{Max Ext VC} and \textsc{Min Ext IS}.
Considering \textsc{Max Ext VC} on $G=(V,E)$ and the particular subset $U=V$ (resp., \textsc{Min Ext IS} with $U=\emptyset$), we obtain
two well known optimization problems  called {\sc upper vertex cover} (\textsc{UVC} for short, also called the {\sc  maximum
minimal vertex cover} problem) and  the {\sc maximum minimal independent set} problem (equivalently \textsc{ISDS} for short). In \cite{Manlove99}, the computational complexity of these two problems have been studied (among 12 problems), and  (in)approximability results are given in \cite{MishraS01,BoriaCP15}
for \textsc{UVC} and in \cite{Halldorsson93a} for \textsc{ISDS} where  lower bounds of $O(n^{\varepsilon-1/2})$ and $O(n^{1-\varepsilon})$,  respectively, for graphs on $n$ vertices are given for every $\varepsilon>0$. Analogous bounds can be also derived depending on the maximum degree~$\Delta$ of the graph. In particular, we deduce:

\begin{corollary}\label{CorInapprox:Min Ext IS}
For any constant $\varepsilon > 0$ and any $\rho\in\Oh\left(n^{1-\varepsilon}\right)$ and $\rho\in\Oh\left(\Delta^{1-\varepsilon}\right)$, there is no  polynomial-time  $\rho$-approximation for \textsc{Min Ext IS} on general graphs of $n$ vertices and maximum degree $\Delta$,  even when $U=\emptyset$, unless $\ptime=\np$. 
\end{corollary}

Now, we strengthen the above mentioned lower bounds of $O(n^{\varepsilon-1/2})$ and  $O(\Delta^{\varepsilon-1/2})$ for the inapproximability of \textsc{Max Ext VC}.

\begin{theorem}\label{Inapprox:Max Ext VC} \SV{$(*)$}
\textsc{Max Ext VC} is as hard as \textsc{MaxIS} to approximate in general graphs even if the set~$U$ of required vertices forms an independent set.
\end{theorem}

\newcommand{\proofofInapproxMaxExtVC}{\noindent Let $S$ be a maximal independent set of $G$ of size $k$; then $S^\prime=\{v \colon v\notin S\} \cup \{v^\prime\colon v\in S\}$ is a  minimal vertex cover of $H$ containing $k$ vertices from $U$. Conversely, let $S^\prime$ be a minimal vertex cover of $H$ extending $U$, with $U^\prime=U\cap S^\prime$. By construction, the set $S^\prime\setminus U^\prime$ is a vertex cover of  $G$ and then $S=V\setminus S^\prime$ is an independent set of $G$ of size $|U^\prime|$. In particular, we deduce $\alpha(G)=opt_{\textsc{Max Ext VC}}(H,U)$.}

\SV{\noindent\emph{Sketch.}}
\LV{\begin{proof}
The proof is based on a simple reduction from \textsc{MaxIS}. 
}Let $G=(V,E)$ be an instance of \textsc{MaxIS}. Construct\LV{ the graph} $H=(V_H,E_H)$ from $G$, where vertex set $V_H$ contains two copies of $V$, \LV{denoted by }$V$ and $V^\prime= \{v^\prime\colon v\in V\}$. \LV{The edge set $E_H$ contains $E$ together with $vv^\prime$ for all $v\in V$, formally}\SV{Let} $E_H=E\cup \{vv^\prime\colon v\in V\}$. Consider $I=(H,U)$ as instance of \textsc{Max EXT VC}, where the required vertex subset is given by $U=V^\prime$.
We claim: $H$ has a minimal vertex cover containing $k$ vertices from $U$ iff $G$ has an independent set of size $k$.
\LV{\proofofInapproxMaxExtVC
\end{proof}}


\SV{\smallskip}
Using the strong inapproximability results for \textsc{MaxIS} given in \cite{Trevisan01,Zuckerman07}, observing $\Delta(H)=\Delta(G)+1$ and
$|V_H|=2|V|$, we deduce the following result.

\begin{corollary}\label{CorInapprox:Max Ext VC}
For any constant $\varepsilon > 0$ and any $\rho\in\Oh\left({\Delta^{1-\varepsilon}}\right)$ and $\rho\in\Oh\left({n^{1-\varepsilon}}\right)$, there is no  polynomial-time $\rho$-approximation for \textsc{Max Ext VC}  on general graphs of $n$ vertices and maximum degree $\Delta$, unless $\ptime=\np$. 
\end{corollary}

\LV{\subsection{Positive results on the price of extension}}
In contrast to the hardness results on these restricted graph classes from the previous sections, we find that restriction to bipartite graphs or graphs of bounded degree improve approximability of \textsc{Max Ext VC}.
For the following results, we assume w.l.o.g. that the input graph is connected, non-trivial and therefore without isolated vertices, as we can solve our problems separately on each connected component and then combine the results. \SV{By simply selecting the side containing the largest number of vertices from $U$, we can show the following.}

\begin{theorem}\label{theo:PoE_VC_Bip} \SV{$(*)$}
A 2-approximation for \textsc{Max Ext VC} on bipartite graphs can be computed in polynomial time.
\end{theorem}
\newcommand{\proofofPoEVCBip}{\begin{proof}
Let $G=((V=V_l\cup V_r),E)$ and $U\subseteq V$ be an instance of \textsc{Max Ext VC}, where $E$ contains only edges connecting $V_l$ and $V_r$.
Since $V_l$ and $V_r$ are both minimal vertex covers ($G$ is without isolated vertices) and also a partition of $V$, then taking one of them containing the largest number of vertices from $U$  (assume it is $V_l$), we get a 2-approximation, because $2 \times |V_l\cap U| \geq |V\cap U|\geq opt_{\textsc{Max Ext VC}}(G,U)$.
\end{proof}}

\LV{\proofofPoEVCBip}

\begin{theorem}\label{theo:PoE-VC_BoundedDegree}
A ${\Delta}$-approximation for \textsc{Max Ext VC} on graphs of maximum degree $\Delta$ can be computed in polynomial time.
\end{theorem}
\begin{proof}
Let $G=(V,E)$ be connected of maximum degree $\Delta$, and $U\subseteq V$ be an instance of \textsc{Max Ext VC}. Consider \LV{the graph $G^\prime=(V^\prime,E^\prime)$ induced  by the open neighborhood of $U$, i.e., $V^\prime=N_G(U)\setminus U$; subgraph}\SV{$G'=G[N_G(U)\setminus U]$. }$G^\prime$ is a  graph of maximum degree (at most) $\Delta-1$ and by Brooks's Theorem, we can color it properly with  at most $\Delta$ colors in polynomial time. Let $(S_1,\dots,S_\ell)$ be
such coloring  of $G^\prime$, with $\ell \leq \Delta$.  Let $U_i=N_G(S_i)\cap U$ for $1\leq i\leq \ell$. $S_i$ is an independent set which dominates $U_i$ in $G$ so it can be extended to satisfy $(iii)$ of Theorem~\ref{caract_Ext_VC}, so $(G,U_i)$ is a \yes-instance of \textsc{Ext VC}.
Choosing $U^\prime=\arg\max  |U_i|$ yields a $\Delta$-approximation, since
 $\Delta \times |S^\prime\cap U^\prime|\geq \sum_{i=1}^\ell |U_i|\geq |U|\geq opt_{\textsc{Max Ext VC}}(G,U)$.
\end{proof}

\noindent
Along the lines of Corollary~\ref{cor:chordal_Ext_VC} with more careful arguments, we can prove:

\begin{theorem}\label{theo:PoE-VC_Chordal} \SV{$(*)$}
\textsc{Max Ext VC} can be solved optimally for chordal graphs and circular-arc graphs in polynomial time.
\end{theorem}

\newcommand{\proofofPoEVCChordal}{
Let $(G,U)$ be an instance of \textsc{Max Ext VC} where $G=(V,E)$ is a  chordal graph (resp., a circular-arc graph)
and $U$ is an independent set. We build a weighted graph $G^\prime$ for \textsc{WMinISDS} such that $G^\prime$ is the subgraph of~$G$ induced by $N_G[U]$ and the weights on vertices are given by $w(v)=1$ if $v\in U$ and $w(v)=0$ for $v\in N_G[U]\setminus U$.
We mainly prove the following equality: 

\begin{equation}\label{eq:PoE-VC_Chordal}
opt_{\textsc{WMinISDS}}(G^\prime,w)=|U|-opt_{\textsc{Max Ext VC}}(G,U)
\end{equation}

\noindent Let  $S^\prime$ be an optimal independent dominating set of $G^\prime$ with total weight $k=opt_{\textsc{WMinISDS}}(G^\prime,w)$; this means that $S^\prime$ contains $k$ vertices from $U$. So, by feasibility of $S^\prime$
and according to Theorem~\ref{caract_Ext_VC}, $U^\prime=U\setminus S^\prime$ is a subset of vertices which can be extended into a feasible solution for \textsc{Max Ext VC} because  $S^\prime\setminus U$ is an independent set which dominates  $U^\prime$. Hence, $opt_{\textsc{WMinISDS}}(G^\prime,w)\geq |U|-opt_{\textsc{Max Ext VC}}(G,U)$. Conversely, let $S$ be a minimal vertex cover
of $G$ such that $opt_{\textsc{Max Ext VC}}(G,U)=|S\cap U|$. By construction, $\textsc{Ext VC}(G,U^\prime)\neq\emptyset$ for $U^\prime=U\cap S$, and then $S^\prime=(N_G[U^\prime])\setminus U) \cap (V\setminus S)$ is an independent set which dominates $U^\prime$  in the subgraph induced by $N_G[U^\prime]$ using Theorem~\ref{caract_Ext_VC}. Hence, by adding at most all vertices of $U\setminus U^\prime$  to $S^\prime$ we get an independent set which dominates all of $U$. Hence, $opt_{\textsc{WMinISDS}}(G^\prime,w)\leq |U|-opt_{\textsc{Max Ext VC}}(G,U)$. In conclusion, equality (\ref{eq:PoE-VC_Chordal}) holds and the result follows.}
\LV{\begin{proof}\proofofPoEVCChordal
\end{proof}}


\section{Bounded degree graphs}
\label{sec:bounded degree}

Our $\np$-hardness results  also work for the case of graphs of bounded degree,  hence it is also interesting to consider 
\textsc{Ext VC}  with standard parameter with an additional degree parameter $d$. 

\begin{theorem}
\textsc{Ext VC} is in $\fpt$ when parameterized both by the standard parameter and by the maximum degree $\Delta$ of the graph.
\end{theorem}

\begin{proof}
For \textsc{Ext VC}, revisit  membership in $\wone$ from\LV{ the proof of} Theorem~\ref{thm: W-c Ext VC}. Guessing a neighbor for each vertex from the given set $U$ is now a limited choice, so that the earlier described algorithm runs in time $\Oh^*(\Delta^k)$. Recursively, the algorithm picks some $u\in U$ and branches on every neighbor $x\in N(u)\setminus U$, considering the new instance $(G-N[x],U\setminus N[x])$. 
Namely, we assume that edges $xu'$ can serve as private edges for all $u'\in U\cap N(x)$. Moreover, $xu$ being a private edge clearly means that $x$ will not be part of the vertex cover (extension). Hence, all neighbors of $x$ (not only the ones already lying in $U$) will be put into the vertex cover extension. 
They all have private edges (namely their connections to $x$). This justifies deleting them from the recursively considered instances.
\end{proof}

Let us look at this algorithm more carefully in the case of $\Delta=3$ analyzing it from the standpoint of exact algorithms, i.e., dependent on the number of vertices $n$ of the graph.
Our algorithm has a branching vector of $(2,2,2)$, resulting in a branching number upper-bounded by 1.733.
However, the worst case is a vertex in $U$  that has three neighbors of degree one. Clearly, this can be improved.
We propose the following reduction rules for  \textsc{Ext VC} on an instance $(G,U)$, $G=(V,E)$, which have to be applied exhaustively and in order:
\begin{enumerate} \setcounter{enumi}{-1} 
    \item If $U=\emptyset$, then answer \yes.
\item If some $u\in U$ is of degree zero, then $(G,U)$ is a  \no-instance.
\item If some $x\notin U$ is of degree zero, then delete $x$ from $V$.
    \item If $u,u'\in U$ with $uu'\in E$, then delete $uu'$ from $E$.
    \item If $u\in U$ is of degree one, then the only incident edge $e=ux$ must be private, hence we can delete $N[x]$ from $V$ and all $u'$ from $U$ that are neighbors of~$x$. 
    \item If $u\in U$ has a neighbor $x$ that is of degree one, then  assume $e=ux$ is the private edge of $u$, so that we can delete $u$ and $x$ from $V$ and $u$ from $U$. 
\end{enumerate}

After executing the reduction rules exhaustively, the resulting graph has only vertices of degree two and three (in the closed neighborhood of $U$) if we start with a graph of maximum degree three. This improves the  branching vector to $(3,3,3)$, resulting in a branching number upper-bounded by 1.443.
However, the rules are also valid for arbitrary graphs, as we show in the following.

\begin{lemma}\label{lem-soundrules} \SV{$(*)$}
The  reduction rules are sound for general graphs
when applied  exhaustively and in order.
\end{lemma}

\newcommand{\proofofLemsoundrules}{\begin{proof}  If $U=\emptyset$, a greedy approach works. 
Clearly, a vertex without neighbors cannot lie in a minimal vertex cover, which justifies Rules~1 and~2.
Likewise, an edge between two vertices in $U$ is clearly covered and cannot serve as a private edge, so that it can be deleted as suggested in Rule~3.
Now, if we consider a vertex $u$ of degree one from $U$ as in Rule~4, then its unique neighbor cannot lie in $U$ (as otherwise this case would have been treated by previous rules). There is simply no choice of a private edge for $u$ but the edge connecting to its unique neighbor $x\notin U$. As $x$ is not going into the vertex cover, all its neighbors must do, but they all have private edges, namely the connections to $x$. This justifies deleting $N[x]$ from $V$ and from $U$, hence Rule~4 is sound.
Finally, a vertex $x$ of degree one that is neighbor of a vertex from $U$ cannot be put into a minimal vertex cover extension, so that the edge connecting it to its unique neighbor $u\in U$ can serve as a private edge for $u$ without hampering any other possibilities. This justifies Rule~5.

As none of the rules will increase vertex degrees, the rules are valid not only in general graphs, but also in graphs of bounded degree.
 \end{proof}}

\LV{\proofofLemsoundrules}

As the previous reasoning was not restricted to $\Delta=3$, we can state:

\begin{theorem}
 \textsc{Ext VC} can be solved in time $\Oh^*((\sqrt[3]{\Delta})^n)$ on graphs of order $n$ with maximum degree $\Delta$. 
\end{theorem}

\noindent
This gives interesting branching numbers for $\Delta=3$: 1.443, $\Delta=4$: 1.588, $\Delta=5$: 1.710, etc., but from $\Delta=8$ on this is no better than the trivial $\Oh^*(2^n)$-algorithm.

Let us remark that the same reasoning that resulted in Rule~5 is valid for\LV{ the following more general version of this rule}:\SV{\mbox{ }\\[-3ex]}

\begin{itemize}
    \item[]\hspace*{-4.6mm}5'\!. If $x\notin U$ satisfies $N(x)\subseteq U$, then delete $N[x]$ from $V$ and from $U$. \SV{\mbox{ }\\[-5ex]}
\end{itemize}

\newcommand{\commentsonreductions}{Namely, as $x$ cannot belong to any minimal vertex cover extension of $U$, all edges incident to $x$ can serve as private edges for the neighbors. Hence, all vertices of $N(x)\cap U$ have received their private edge and can hence be removed.

Notice that in order to decide if $(G,U)$ is a \yes-instance of \textsc{Ext VC}, it is sufficient to restrict our attention to $G'=G[N_G[U]]=(V',E')$ according to Theorem~\ref{caract_Ext_VC}.
Hence, every vertex in $V'=N_G[U]$ either belongs to $U$ or is neighbor of some vertex in $U$. 
We can formulate this as a reduction rule as follows.\SV{\mbox{}\\[-8ex]}\LV{\commentsonreductions}}

\begin{enumerate} \setcounter{enumi}{5} 
    \item Delete $V\setminus N_G[U]$.\SV{\mbox{ (inspired by Theorem~\ref{caract_Ext_VC})}\\[-5ex]}
\end{enumerate}

%

We now run the following branching algorithm:
\begin{enumerate}
    \item Apply all reduction rules exhaustively in the order given by the numbering.
    \item On each connected component, do:
    \begin{itemize}
        \item Pick a vertex $v$ of lowest degree.
        \item If $v\in U$: Branch on all possible private neighbors.
        \item If $v\notin U$: Branch on if $v$ is not in the cover or one of its neighbors.
    \end{itemize}
\end{enumerate}

A detailed analysis of the suggested algorithm gives the following result.

\begin{theorem}\label{thm-ExtVCsubcubic} \SV{$(*)$}
\textsc{Ext VC} on subcubic graphs can be solved in time $\Oh^*(1.26^n)$ on graphs of order $n$.
\end{theorem}

\newcommand{\proofofExtVCsubcubic}{
\begin{proof}
We prefer branching on vertices of degree two as follows (after always executing all reduction rules first; this also means that the graph has minimum degree of two). Consider $v$ with  $N_{G'}(v)=\{x,y\}$.
\begin{itemize}
    \item If $v\notin U$, then, w.l.o.g., $x\in U$ but $y\notin U$, as clearly one of the neighbors of $v$ must belong to $U$, but both cannot because of Rule~5'. Now, as not all vertices of $N_{G'}[v]$ can belong to the cover because of minimality, so that either $v$ is not in the cover or $y$ is not in the cover. In either case, we can remove at least three vertices from the graph, resulting in a branching vector of $(3,3)$.
    \item If $v\in U$, then (by Rule~3),
    $x,y\notin U$. Either $xv$ or $yv$ must serve as a private edge, i.e., $x$ or $y$ does not belong to the cover, but all of its neighbors. A simple analysis gives again a branching vector of $(3,3)$.
\end{itemize}

If we prioritize the branching on $v\notin U$, then we know that both neighbors of any $v\in U$ that we branch on have degree three, which improves the branching vector to $(4,4)$ in that case.

After this branching, we would arrive at a 3-regular graph. However, as all branching rules or also the reduction rules always remove vertices in the neighborhood of some vertex, such 
a cubic graph can never result.
Hence, such a possibly bad branching can only occur at the very beginning, but this can be ignored in the running time analysis.
\end{proof}
}

\LV{\proofofExtVCsubcubic}

\begin{corollary}\label{cor-ExtVCsubcubic}  \SV{$(*)$}
\textsc{Ext VC} on subcubic graphs can be solved in time $\Oh^*(2^{|U|})$ with fixed vertex set $U$.
\end{corollary}

\newcommand{\proofofCorExtVCsubcubic}{\begin{proof}
Let us analyze  the same branching algorithm can be also analyzed from the viewpoint of parameterized complexity. 
In each branching step, at least one vertex from $U$ is removed.
\end{proof}
}

\LV{\proofofCorExtVCsubcubic}

\LV{Finally, we can see that o}\SV{O}ur reduction rules guarantee that each vertex not in $U$ (and hence in $N_G(U)$) has one or two neighbors in $U$, and \LV{similarly, }each vertex in $U$ has two or three neighbors in $N_G(U)$.
Hence, $|N_G(U)|\leq 3|U|$. In general, due to Rule~6:

\begin{theorem}
\textsc{Ext VC} on  graphs of maximum degree $\Delta$ allows for a vertex kernel of size $(\Delta+1)|U|$, parameterized by the size of the given vertex set $U$.
\end{theorem}

Looking at the dual parameterization (i.e., \textsc{Ext IS} with standard parameter), we can state due to all reduction rules:

\begin{theorem}
\textsc{Ext VC} on  graphs of maximum degree $\Delta$ allows for a vertex kernel of size $\frac{\Delta-1}{2}|V\setminus U|$, parameterized by \LV{the size of the complement of the given vertex set $U$}\SV{$|V\setminus U|$}.
\end{theorem}

For $\Delta=3$, we obtain vertex kernel bounds of $4|U|$ and $2|V\setminus U|$, respectively. 
With the computations of \cite[Cor. 3.3 \& Cor. 3.4]{ChenFKX07}, we can state the following\LV{ result}.

\begin{corollary}
Unless $\ptime=\np$, for any $\varepsilon>0$, there is no   size $(2-\varepsilon)|U|$ and no size $(\frac{4}{3}-\varepsilon)|V\setminus U|$ vertex kernel for \textsc{Ext VC} on  subcubic graphs, parameterized by $|U|$ or $|V\setminus U|$, respectively.
\end{corollary}
\noindent
This shows that our (relatively simple) kernels are quite hard to improve on.

\newcommand{\remarksonplanarremarks}{Let us exemplify this observation with the example of planar graphs, somehow complementing the results stated in Theorem~\ref{thm-fpt-VC-planar}.
Observe that all reduction rules that we have put up preserve planarity.
So, we are left with a planar graph $G$ together with a dominating set of required vertices $U$.
The branching is then done as follows. Pick some vertex $x$ of smallest degree. As $G$ is planar, $x$ has degree at most five. If $x\in U$, then one of its incident edges must by private to $x$, leading to a branching number of 5, as at least one of the vertices of $U$ (namely, $x$) is removed from the instance in each branch.  If $x\notin U$, then either $x$ will be put into the minimal vertex cover extension or not.
If not, then all of the neighbors of $x$ must be in the cover; among these, at least one belongs to $U$, see Theorem~\ref{caract_Ext_VC} $(iii)$. If $x$ comes into  the minimal vertex cover extension, then not all of its neighbors could be there (by minimality). Hence, we can branch on the at most five neighbors which of them provide a private edge to $x$. Clearly, at least one of the neighbors is in $U$, so that this cannot give a private edge. As argued before, at least one of the other neighbors of $v$ (in the branch when $xv$ is assumed to be a private edge of $x$) must lie in $U$, so that the parameter reduces by one in each branch. Altogether, this branching has again a branching number of 5, finally leading to an algorithm with running time $\Oh^*(5^{|U|})$.  Further improvements on this branching should be possible by using \cite[Theorem~2]{AksionovBMSST05}, along the lines of thinking elaborated in~\cite{AbuFerLan2008} for a related problem on planar graphs.
When we analyze the mentioned algorithm as an exact algorithm, always branching on vertices of smallest degree, which is at least two by our reduction rules, 
depending on the number $n$ of vertices, branching vectors of $(3,3)$, $(4,4,4)$, $(5,5,5,5)$, 
$(6,6,6,6,6)$ (or better) would result, yielding a branching number of 1.32.

Finally, observe that this approach does not translate to graphs of bounded average degree, as such class of graphs is not closed  under taking induced subgraphs.}

\begin{remark} \label{rem-planar} \SV{$(*)$}
Note that the arguments that led to the $\fpt$-result for \textsc{Ext VC} on graphs of bounded degree (by
providing a branching algorithm) also apply to graph classes that are closed under taking induced subgraphs and that guarantee the existence of vertices of small degree. 
\SV{This idea leads to a branching algorithm with running time $\Oh^*(5^{|U|})$ or  $\Oh^*(1.32^{|V|})$.}\LV{\remarksonplanarremarks}
\end{remark}

\newcommand{\remarksontrees}{There is also a linear-time algorithm that completely avoids using DP. This algorithm makes use of the reduction rules presented above and 
works as follows.

By Rule~6, based on Theorem~\ref{caract_Ext_VC}, we can restrict our attention to tree instances $(T,U)$ such that all vertices not in $U$ are neighbors of vertices in $U$ and either $U$ has at most one vertex, or every vertices of $U$ has another vertex of $U$ at distance at most three. If $(T,U)$ does not satisfy this requirement motivated by Theorem~\ref{caract_Ext_VC}, we can easily produce a number of instances that satisfy these requirements, such that the original instance is a \yes-instance iff each of the generated instances is a \yes-instance. In addition, Rule~3 guarantees that the set $U$ forms an independent set in $T$.

Therefore, our reduction rules will resolve such instances completely, because either $T$ has only one vertex, in which case $(T,U)$ is a \yes-instance if and only if $U=\emptyset$ according to Rules~0, 1 and~2, or $T$ has  a vertex~$x$ of degree one. In that case, either Rule~4 applies, namely, if $x\in U$, or $x\notin U$, which means that there is a vertex $u$ of $U$ that is neighbor of $x$ (due to Rule~6). Hence, Rule~5 applies.

In this way, the whole tree will be finally processed, because each rule application will delete certain  vertices.

In particular, applying these rules could also split-up the tree, in which case the components can be separately processed.}

\begin{remark} 
\label{rem-trees}  \SV{$(*)$}
Let us mention that we also derived several linear-time algorithms for solving \textsc{Ext VC} (and hence \textsc{Ext IS}) on trees in this paper.
\LV{\smallskip
\noindent
\underline{Algorithm 1}:}\SV{(1)} A simple restriction of the mentioned  DP algorithm on graphs of bounded treewidth solves this problem. \LV{The main drawback is that this algorithm is rather implicit in this paper.

\smallskip
\noindent
\underline{Algorithm 2}:}\SV{(2) Apply our reduction rules exhaustively.}\LV{\remarksontrees}
(3) Check the characterization given in Theorem~\ref{Tree_carac_Ext_VC}. Also, Theorem~\ref{theo:PoE-VC_Chordal} provides another polynomial-time algorithm on trees.
\end{remark}

\section{Conclusions}\label{sec:conclusions}

We have found many graph classes where \textsc{Ext VC} (and hence also \textsc{Ext IS}) remains $\np$-complete, but also many classes where these problems are solvable in poly-time. The latter findings could motivate looking into parameterized algorithms that consider the distance from  favorable graph classes in some way.

It would be also interesting to study further optimization problems that could be related to our extension problems, for instance the following ones, here formulated as decision problems \SV{(a)}\LV{\begin{itemize}
    \item} Given $G,U,k$, is it possible to delete at most $k$ vertices from the graph such that $(G,U)$ becomes a \yes-instance of \textsc{Ext VC}?
    
    Clearly, this problem is related to the idea of the price of extension discussed in this paper, in particular, if one restricts the possibly deleted vertices to be vertices from $U$.
\LV{\item}\SV{(b)}  Given $G,U,k$, is it possible to add at most $k$ edges from the graph such that $(G,U)$ becomes a \yes-instance of \textsc{Ext VC}?
\LV{
}  Recall that adding edges among vertices from $U$ does not change our problem, as they can never be private edges, but adding edges elsewhere might create private edges for certain vertices.
\LV{\end{itemize}
}
Such problems would be defined according the general idea of graph editing problems studied quite extensively in recent years. 
These problems are particularly interesting in graph classes where \textsc{Ext VC} is solvable in poly\SV{-}\LV{nomial }time.

Considering the underlying classical optimization problems, it is also a rather intriguing question, whether it is possible to decide for a given set $U$ if it can be extended not just to any inclusion minimal vertex cover but to a globally smallest one. It is tempting to think that our results on hardness of minimal extension transfer to this variant, but this is not the case. However, clearly, this is also a very difficult but rewarding task, as an efficient way to answer this question would yield an efficient algorithm to solve the minimum vertex cover problem.
From an approximation point of view, one might want to discuss how much it costs to force a fixed subset $U$ to be in the solution, in case of a scenario where this set is fixed for some reason, i.e., how large a minimal solution extending $U$ is compared to the optimum vertex cover which is not forced to include $U$. The simple example of a vertex $v$ of degree $n$ with $n$ neighbors of degree~1 gives an arbitrarily bad gap for this kind of consideration; the cardinality of an optimal vertex cover is~1 (the set $\{v\}$ is a minimum vertex cover), while a minimal solution extending the set $U=\{w\}$, where $w$ is a degree 1 neighbor of $v$, has cardinality $n$. At last, one could discuss how much of a subset $U$ can be kept when aiming for not just a minimal solution, as we did here with our discussion on the price of extension, but for a minimum one. 

\LV{\bibliographystyle{plain}}\SV{\bibliographystyle{abbrv}}
\bibliography{biblio}

\newpage

\section{Generalizations to Extensions of $H$-graph cover and $H$-free subgraph}\label{sec:hfree}

Assume that graph $H=(V_H,E_H)$ is fixed; the {\sc maximum induced $H$-free subgraph problem}, \textsc{Induced $H$-free} for short, asks, given a graph $G=(V,E)$, to find the largest subset of vertices $S\subseteq V$ such that the subgraph $G[S]$ induced by $S$ 
is $H$-free, i.e., it does not contain any copy of $H$. A corresponding extension version is given by:

\begin{center}
\fbox{\begin{minipage}{.95\textwidth}
\noindent{\textsc{Ext Induced $H$-free}}\\\nopagebreak
{\bf Input:} A  graph $G=(V,E)$, a set of vertices $U \subseteq V$. \\\nopagebreak
{\bf Question:}  Does $G$ have a maximal subgraph $G[S]$ induced by $S$  with $S\subseteq U$ which is $H$-free?
\end{minipage}}
\end{center}

The particular case of $H=K_2$ corresponds to \textsc{Ext IS}, because $S$ induces a  $K_2$-free subgraph iff it is an independent set. We now generalize our previous findings on the complexity of  \textsc{Ext IS} towards this more general setting. Recall that a graph is \emph{biconnected} if it stays connected after deleting any single vertex.
\begin{theorem}\label{Theo:Ext_H-free}
If $H$ is biconnected with at least 2 vertices, then \textsc{Ext Induced $H$-free} is $\np$-complete.
\end{theorem}
\begin{proof}
Let $H=(V_H,E_H)$ with $n_H=|V_H|$ vertices be a biconnected graph and assume $n_H\geq 3$ ( $n_H=2$ corresponds to \textsc{Ext IS} which has been proved $\np$-complete in Theorem \ref{Bip_Ext_VC}). The proof is based on a reduction
from {\sc  $H$-free 2-colorability}, denoted by 
\textsc{$H$-2Col} for short.
With fixed $H$, the problem 
\textsc{$H$-2Col}
consists in deciding if the vertices of a given graph $G=(V,E)$ can be partitioned into two $H$-free induced subgraphs $G_{V_i}$,  $i=1,2$ (so $V=V_1\cup V_2$). 
\textsc{$H$-2Col} is $\np$-complete iff $H$ contains at least 3 vertices; see~\cite{Achlioptas97}.

From $G=(V,E)$ with $n$ vertices, as an  instance of \textsc{$H$-2Col},
we build an instance of \textsc{Ext Induced $H$-free} as follows:

Let $u$ and $w$ be two distinct vertices of $H$. We consider two copies $G_1$ and $G_2$ of $G$ where $v^1$ and $v^2$ are copies of vertex $v\in V$ and $2n$ copies $H_i$ of $H$ (where $u_i$ and $w_i$ are copies of $u$ and $w$). We collapse together two copies $H_i$ and $H_{n+i}$ by merging vertices $w_i$ and $w_{n+i}$; Let $H^\prime_i$ be the resulting graph and $w^\prime_i$ be the vertex corresponding to$w_i$ and $w_{n+i}$ after the merging. Now, we merge vertices $v^1_i$ with $u_i$ and $v^2_i$ with $u_{n+i}$ and we get the graph $G^\prime$ as part of an instance of \textsc{Ext Induced $H$-free}. Hence, $G^\prime$ contains $G_1$, $G_2$ and the graphs  $H^\prime_i$ for $i=1,\dots,n$. Finally, we set $U=V(G^\prime)\{w^\prime_i\colon 1\leq i\leq n\}$.

We claim that $G$ is a yes-instance of \textsc{$H$-2Col} 
 iff  ($G^\prime$,$U$)
is a yes-instance of \textsc{Ext Induced $H$-free}.

Assume that $(V_1,V_2)$ is an $H$-free 2-coloring (bipartition) 
 of $G$. Consider  any maximal $H$-free subgraph of $G$ containing
$V_1$ (resp., $V_2$) and call it $V^\prime_1$ (resp., $V^\prime_2$). Finally, let us denote by $V(H^\prime_i)$ the set of vertices of subgraph $H^\prime_i$ for each $i\in\{1,\dots,n\}$.
We claim that the set $S=\bigcup_{i=1}^n (V(H^\prime_i)\setminus \{w^\prime_i,v_i^1,v_i^2\})\cup \bigcup_{j=1}^2\{v_i^j: v_i\in V^\prime_j\}$  induces a maximal $H$-free subgraph of $G^\prime$ with $S\subseteq U$. Actually, this is clear inside either each copy of $G$ or each $V(H^\prime_i)\setminus \{w^\prime_i\}$. An assumed copy $R$ of $H$ in $S$ must hence include vertices in a copy $G_j$ with $j=1,2$ and also vertices in a copy $H^\prime_i$ for some $i=1,\dots,n$.
Hence, to be connected, $R$ has to contain some $v_i^j$ which is a cut-vertex of $R$, separating the vertices in $R$ from $V(G_j)\setminus\{v_i^j\}$, for some $j\in \{1,2\}$, from $V(H_i^\prime)\setminus\{v_i^j\}$, which is a contradiction to $H$ being biconnected.

Conversely, assume that there exists a set $S\subseteq U$ which induces a maximal  $H$-free subgraph of $G$. Let $S_i$ for $i=1,2$ be the vertices of $S$ included in copy $G_i$.
Let $V_1=\{v_j\colon v_j^1\in S_1\}$ and $V_2=\{v_j\colon v_j^2\in S_2\mbox{~and~} v_j^1\notin S_1\}$. We claim that  $(V_1,V_2)$ is an $H$-free 2-coloring 
of $G$. Obviously, each subgraph of $G$ induced by $V_i$ is $H$-free for $i=1,2$. If  $(V_1,V_2)$ is not a vertex partition of $G$, then there exists some $i\in \{1,\dots,n\}$ such that $v_i^j\notin S$ for both $j=1$ and $j=2$ for some  $i=1,\dots,n$. This however implies that  $(V(H^\prime_i)\setminus \{v_i^1,v_i^2\})\subseteq S$ because
$(V(H^\prime_i)\setminus \{v_i^1,v_i^2\})$ is $H$-free (recall that $H$ is assumed to be biconnected). Thus, $w^\prime_i\in S$ which is a contradiction to $S\subseteq U$.
\end{proof}

We are now considering a covering analogue to \textsc{Ext Induced $H$-free}.

\begin{center}
\fbox{\begin{minipage}{.95\textwidth}
\noindent{\textsc{Ext $H$-cover}}\\\nopagebreak
{\bf Input:} A  graph $G=(V,E)$, a set of vertices $U \subseteq V$. \\\nopagebreak
{\bf Question:}  Does $G$ have a minimal subset $S$ which covers all copies of $H$ with $U\subseteq S$?
\end{minipage}}
\end{center}

Similarly to previous remark, we have $\textsc{Ext $K_2$-cover}=\textsc{Ext VC}$. More generally for any fixed graph $H$, $S$ is a minimal $H$-cover of $G=(V,E)$ iff $V\setminus S$ is a maximal $H$-free subgraph. Hence, using same the reasoning that the one given in Remark \ref{rem-Ext VC=IS}, we deduce:
\begin{proposition}\label{prop:Ext_H-cover}
If $H$ is biconnected, then \textsc{Ext $H$-cover} is $\np$-complete.
\end{proposition}

Notice that this last assertion is interesting, as for the corresponding classical problem \textsc{$H$-cover}, no easy conditions like biconnectivity are known to yield $\np$-completeness results, see\SV{~\cite{KratochvilPT98}}\LV{ \cite{AbeFelSti91,KratochvilPT97,KratochvilPT98}}.

We are now stating a characterization of graphs admitting an $H$-cover extension that could be compared to Theorem~\ref{caract_Ext_VC}.

\begin{theorem}\label{caract_Ext_H_cover}
Let $G=(V,E)$ be a graph and  $U \subseteq V$ be a set of vertices. There is a minimal $H$-cover $S$ of $G$ extending $U$ iff the two
following conditions hold:

\begin{description}
\item[$(i)$] For every $u\in U$, there is a copy $H_u=(V(H_u),E_u)$ of $H$ in $G$ such that $V(H_u)\cap U=\{u\}$.

\item[$(ii)$] If $V'=\bigcup_{u\in U}V(H_u)$, then the subgraph $G'$ of $G$ induced by $V'\setminus U$ is $H$-free.
\end{description}
\end{theorem}
\begin{proof}
Let $G=(V,E)$ be a graph and  a set $U \subseteq V$ be a set of vertices. The condition is sufficient; indeed,
let $S'$ be any maximal $H$-free subgraph containing $V'$ where $V'$ is defined according to conditions $(i)$ and $(ii)$ (via $H_u$).
The set $V\setminus S'$ is a minimal $H$-cover of $G$ extending $U$. 

Conversely, assume that  $S$ is a minimal $H$-cover of $G$ extending $U$.
Let us prove that $S$ satisfies conditions $(i)$ and $(ii)$. Since $U\subseteq S$ is a minimal $H$-cover of $G$, then for every
$u\in U$, there exists a copy $H_u=(V(H_u),E_u)$ of $H$ in $G$ such that $u\in V(H_u)$ and $S\setminus\{u\}$ does not cover $H_u$. In particular,
we deduce  $V(H_u)\cap U=\{u\}$. Now, let $V'=\bigcup_{u\in U}V(H_u)$; if the subgraph $G'$ of $G$ induced by $V'\setminus U$ is not $H$-free,
then $\exists v\in (V'\cap S)\setminus U$, such that  $v$ lies in  $V(H_{u_0})$ for some $u_0\in U$,   contradicting the fact that
$S\setminus\{u_0\}$ does not cover $H_{u_0}$.
\end{proof}

\begin{corollary}\label{cor:ExtH_Cover_XP}
For every fixed $H$, \textsc{Ext $H$-cover} parameterized by $|U|$ is in~$\xp$.
\end{corollary}
\begin{proof}
Using exhaustive search, finding $V'=\bigcup_{u\in U}V(H_u)$ can be done in  time $\Oh(n^{k+n_H})$, where $k=|U|$ and $n_H=|V(H)|$. The remaining steps can be performed in $\Oh(n^{k+n_H})=\Oh(n^k)$ time, as $H$ is fixed.
\end{proof}

\end{document}